\documentclass{article}
\usepackage{graphicx} 

\usepackage{dsfont}

\usepackage{psfrag}

\usepackage{calrsfs}
\usepackage{mathrsfs}
\usepackage{pdfsync}
\usepackage{amsmath, amsthm, amssymb, amsfonts}
\usepackage[shortlabels]{enumitem}

\usepackage{url}
\usepackage{float}

\usepackage[usenames]{color}
\usepackage{fullpage}

\usepackage{tikz}

\usetikzlibrary{arrows,decorations.pathmorphing,backgrounds,positioning,fit,automata}

\usetikzlibrary{arrows}
\usetikzlibrary{petri}
\usetikzlibrary{topaths}

\usepackage{indentfirst,calc,euscript}
\usepackage{setspace}
\usepackage[reals]{layout}
\usepackage{xr}
\usepackage{amscd}

\usepackage{sgame}
\usepackage{subfigure}

\usepackage[colorlinks=true,breaklinks=true,bookmarks=true,urlcolor=blue,
     citecolor=blue,linkcolor=blue,bookmarksopen=false,draft=false]{hyperref}
\usepackage{cleveref}

\newcommand{\ignore}[1]{}

\newtheorem{theorem}{Theorem}

\newtheorem{lemma}[theorem]{Lemma}

\newtheorem{proposition}[theorem]{Proposition}

\theoremstyle{definition}

\numberwithin{equation}{section}

\numberwithin{theorem}{section}

\newcommand{\m}{\mathbb}

\title{Prophet Inequalities Require Only a Constant Number of Samples}
\author{
    Andrés Cristi \thanks{CMM, Universidad de Chile, Chile}
    \and 
    Bruno Ziliotto\thanks{CEREMADE, CNRS, Paris Dauphine University, France, and CMM, Universidad de Chile, Chile}
}
\date{}
\begin{document}

\maketitle
\begin{abstract}
    In a prophet inequality problem, $n$ independent random variables are presented to a gambler one by one. The gambler decides when to stop the sequence and obtains the most recent value as reward. We evaluate a stopping rule by the worst-case ratio between its expected reward and the expectation of the maximum variable. In the classic setting, the order is fixed, and the optimal ratio is known to be 1/2. Three variants of this problem have been extensively studied: the prophet-secretary model, where variables arrive in uniformly random order; the free-order model, where the gambler chooses the arrival order; and the i.i.d. model, where the distributions are all the same, rendering the arrival order irrelevant. 
    
   Most of the literature assumes that distributions are known to the gambler. Recent work has considered the question of what is achievable when the gambler has access only to a few samples per distribution. Surprisingly, in the fixed-order case, a single sample from each distribution is enough to approximate the optimal ratio, but this is not the case in any of the three variants. 
   We provide a unified proof that for all three variants of the problem,  a constant number of samples (independent of $n$) for each distribution is good enough to approximate the optimal ratios. 
   Prior to our work, this was known to be the case only in the i.i.d. variant. Previous works relied on explicitly constructing sample-based algorithms that match the best possible ratio. 
   Remarkably, the optimal ratios for the prophet-secretary and the free-order variants with full information are still unknown. Consequently, our result requires a significantly different approach than for the classic problem and the i.i.d. variant, where the optimal ratios and the algorithms that achieve them are known. We complement our result showing that our algorithms can be implemented in polynomial time.
   
   A key ingredient in our proof is an existential result based on a minimax argument, which states that there must exist an algorithm that attains the optimal ratio and does not rely on the knowledge of the upper tail of the distributions. A second key ingredient is a refined sample-based version of a decomposition of the instance into ``small" and ``large" variables, first introduced by Liu et al. [EC'21]. 
   The universality of our approach opens avenues for generalization to other sample-based models.
   Furthermore, we uncover structural properties that might help pinpoint the optimal ratios in the full-information cases. 
\end{abstract}

\section{Introduction}

The \textit{Prophet Inequality} is a fundamental problem in optimal stopping theory, in which a gambler is successively proposed with $n$ realizations of positive independent random variables and has to pick one of them. The gambler knows in advance the order and the distribution of each variable but upon observing each realization must decide irrevocably whether to pick it. A classic result by Krengel and Sucheston~\cite{KS77} asserts that the gambler can get at least half of the expected maximum of the variables, and that this is the best possible guarantee that is independent of the variables' distributions. Remarkably, Samuel-Cahn~\cite{S84} proved this can be achieved using a very simple rule: pick any variable that is above the median of the distribution of the maximum. In the last decade, due to its connections with mechanism design and posted price mechanisms \cite{HKS07,CHM10,CFPV19}, the prophet inequality and its many variants have become an intensely studied topic and a staple framework to study online selection problems beyond worst-case analysis.

Three variants of this problem have been extensively studied.
First, the \textit{i.i.d. problem}, in which variables have i.i.d. distributions. There, the optimal ratio is $\beta \simeq 0.745$, where $1/\beta$ is the unique solution of
$\int_0^1 \frac{1}{y(1-\ln(y))+(\beta-1)} dy =1$. The upper bound was shown in \cite{HK82,kertz86}, and the lower bound in \cite{CFHO17}. 
Second, the \textit{Prophet Secretary problem}, in which variables appear in uniformly random order. Esfandiari et al.~\cite{EHL15} initiated the study of this variant, showing that the gambler can guarantee a factor of $1-1/e$, and later Ehsani et al.~\cite{EHK18} showed this can be achieved with a single-threshold rule.  Azar et al.~\cite{ACK18} slightly improved the $1-1/e$ factor by using a multi-threshold algorithm, and then Correa et al.~\cite{CSZ21} proved the optimal factor lies in $[0.669,0.732]$. The current known best upper bound is $0.724$ \cite{GMS23,BC23}, and it remains one of the most important open problems in the area to close this gap. 
Last, in the \textit{Free-order problem}, variables are ordered by the gambler. The best-known upper bound is the i.i.d. model ratio $1/\beta$. Lower bounds have been successively obtained by \cite{ACK18,CHM10,CSZ21}, and huge progress was made quite recently by Peng and Tang~\cite{PT22}, who established a lower bound of $0.724$, which was later improved to $0.725$~\cite{BC23}.  

In parallel, an exciting recent line of work has considered the more realistic case where the gambler does not have full access to the distributions, but instead observes samples from past data beforehand. Rubinstein, Wang and Weinberg~\cite{RWW20} showed that a single sample per distribution is enough to achieve the best possible factor of $1/2$ in the classic prophet inequality. Moreover, they prove that in the i.i.d. case, $O(1/\varepsilon^6)$ are enough to achieve the best possible guarantee of $0.745-\varepsilon$. Recently, Correa et al.~\cite{CCES23} showed that $O(1/\varepsilon)$ are enough to guarantee $0.745-\varepsilon$. Correa et al.~\cite{CCES22} showed that in the prophet secretary problem, one sample per distribution is enough to guarantee a factor of $0.635$. 

The focus of our work is on sample-based versions of the Prophet Secretary problem and of the Free-Order problem.
In both models, our main question is what fraction of the expected maximum can be guaranteed using a constant (independent of $n$) number of samples per distribution.

\subsection{Our result and technical highlights}
Let $C_S$ be the optimal fraction of the expected maximum that can be guaranteed in the prophet secretary problem. 
We prove that for any $\varepsilon>0$, it is possible to guarantee a $C_S-\varepsilon$ factor in the sample-based prophet secretary problem, using no more than $O(1/\varepsilon^5)$ samples from each distribution. The exact same result holds for the sample-based free-order problem, with the corresponding optimal ratio. 
Our proof is ``universal'', in the sense that it deals simultaneously with both models, and also works for the i.i.d. model. 

Analogous results for the prophet inequality and the i.i.d. prophet inequality rely on either converting an existing algorithm with the optimal guarantee into a sample-based one, or on constructing a sample-based algorithm and showing it matches the best-possible guarantee. Remarkably, since the best-possible guarantee for the prophet secretary problem and the free-order problem are unknown, such approaches cannot be used to show our result, and instead, we establish new properties of the problem. Moreover, the optimal algorithms for the classic and the i.i.d. variants use no more than $n$ thresholds, one for each variable. In contrast, in the random order case, the optimal algorithm uses an exponential number of thresholds, one for each variable and each possible arrival order. Similarly, the optimal algorithm for the free-order model has to choose among the exponentially many arrival orders.
\\

Before describing the main lines of the proof, let us highlight the difficulty of proving the result with an example in the prophet-secretary variant. First, consider the instance $(X_1,\dots,X_{n})$, such that $X_1,\dots,X_{n-1}$ are i.i.d. and equal to $n$ with probability $n^{-2}$, and 0 otherwise. The variable $X_n$ is deterministic, equal to $\sqrt{3}-1$. Assume that the gambler knows the distributions. This corresponds to the example in \cite{CSZ21}, where it is shown that the gambler cannot guarantee a ratio better than $\sqrt{3}-1+o(1)$, which proves that $C_S \leq \sqrt{3}-1$. 

Now, consider the following other problem: given a positive number $a$, $X_1,\dots,X_{n-1}$ are i.i.d. and equal to $a \cdot n$ with probability $n^{-2}$, and 0 otherwise. The variable $X_n$ is deterministic, equal to $\sqrt{3}-1$. The number $a$ is unknown to the gambler, who has access to a constant number of samples of each distribution. For $n$ large, with probability at least $1-O(1/n)$, the samples of $X_1,\dots,X_{n-1}$ are all equal to 0, hence uninformative. Hence, this problem is seemingly much harder than the previous one, and one may expect that the ratio guaranteed by the gambler goes way below $\sqrt{3}-1$, possibly below $C_S$. Our result shows that it is not the case: the gambler can still guarantee $C_S$. Surprisingly, one of the proof steps shows that he can even guarantee $\sqrt{3}-1$: hence, when $a$ is adversarially chosen, knowing $a$ or not knowing $a$ does not change the guarantee. 
\\

Our proof consists of three main steps, which are, to some extent, important facts about the prophet-secretary and the free-order variants by themselves.

\textbf{Step 1} of our proof is to show that essentially we do not need to know the upper tails of the distributions in order to achieve the best-possible guarantee. This alleviates a heavy burden on the design of sample-based algorithms, as the upper tails potentially contribute most of the expectation of the maximum, and precisely estimating them might require an arbitrary high number of samples. The proof of this fact is based on a minimax argument: if by observing the upper tails of the distributions we can design an algorithm that guarantee the optimal constant, by choosing a randomized algorithm, we can also guarantee the optimal constant against an adversary that decides how large is the contribution of the upper tail of each distribution to the expected maximum.

\textbf{Step 2} relies on the notion of $\varepsilon$-small distributions, introduced by Liu et al.~\cite{LPP21}. A variable is $\varepsilon$-small if the probability that it is larger than zero is at most $\varepsilon$. Liu et al. show that in the prophet secretary problem, if all variables are $\varepsilon$-small, it is possible to guarantee a fraction of $0.745$ of the expected maximum, which is the best possible guarantee also if the variables are i.i.d. Our result in this step is to show that if a large proportion of the variables are $\varepsilon$-small, then we can pretend those variables are i.i.d. by losing only an $\varepsilon$ fraction of the expected maximum. The main idea is to show that for a fixed algorithm, replacing the $\varepsilon$-small variables with i.i.d. variables in a way that does not change the distribution of the maximum, we stop the sequence only earlier, and conditional on stopping with an $\varepsilon$-small variable, its expectation is almost the same as if the $\varepsilon$-small variables were i.i.d.

In \textbf{Step 3}, we show how to actually use the samples to construct the algorithm. We further divide step 3 into step 3(a) and step 3(b). In step 3(a) we show that using constantly many samples per distribution, we can split the set of variables into two sets, one containing at least $n-O\left( (1/\varepsilon)\log(1/\varepsilon) \right)$ $\varepsilon$-small variables. Because of step 2, we can replace this large set of variables with i.i.d. variables. In step 3(b), we show that using constantly many samples per distribution, we can estimate very well the distribution of the auxiliary i.i.d. variables, as well as the distribution of the constantly-many variables that are not $\varepsilon$-small, except for their upper tails.

Finally, notice these three steps alone only guarantee the existence of a sample-based algorithm. In fact, step 1 uses a minimax argument that is non-constructive. We complement this by describing in \textbf{Step 4} a procedure that finds such an algorithm and runs in polynomial time. The starting point is a linear program of exponential size that captures the algorithm from step 1. We show how to reduce the linear program to one of polynomial size by leveraging the fact that we are only interested in solving instances where all variables have supports of polynomial size, and most of them are i.i.d.

\subsection{Further related work}

The framework of the prophet inequality has been generalized to a wide variety of online selection problems beyond single selection. Important generalizations include prophet inequalities for $k$-selection~\cite{CDL23,JMZ22}, matroid and matroid intersection~\cite{KW12,SVW22}, matching~\cite{AHL12,GW19}, and online combinatorial auctions~\cite{FGL14,CC23}. In these generalizations, the gambler can select multiple variables under some combinatorial constraint on the selected set, instead of just one.

Pioneered by Azar, Kleinberg and Weinberg~\cite{AKW14}, several recent works study the question of what guarantees are possible in prophet inequality models under limited sample access to the distributions.
Azar et al.~\cite{AKW14} showed that there was a connection between this model and the secretary problem, as many algorithms for the secretary problem can be adapted to obtain constant-factor sample-based prophet inequalities. Caramanis et al.~\cite{CDF22} consider sample-based greedy algorithms, which are, in a sense, a refinement of the framework of Azar et al~\cite{AKW14}. With this framework, they obtained improved factors for various classes of matroids.

For the case of selecting a matching on a graph, where edges have random weights, Duetting et al.~\cite{dutting2021prophet} and Kaplan, Naori and Raz~\cite{KNR22} recently considered the case where the gambler has a single sample of each edge beforehand and showed constant-factor approximations in edge-arrival and vertex-arrival models.

For the case of combinatorial auctions, where the gambler is a seller with a set of items for sale and the random variables correspond to the valuation functions of buyers, Feldman et al.~\cite{FGL14} and Correa et al.~\cite{CCF23}, besides showing approximation factors for the full-information case, gave sample-based versions, using polynomially many samples per distribution and assuming bounded supports.

Gravin et al.~\cite{GLT22} recently studied the prophet inequality with less than one sample per distribution, i.e., we have a sample from each distribution with probability $p$ independently, in the classic fixed order version. They showed that this model smoothly interpolates between a guarantee of $0$ if there are no samples, and the guarantee of $1/2$ if we have one sample per distribution. Similarly, Correa et al.~\cite{CCES23} considered a similar question for the i.i.d. variant, where the gambler has access to $n\cdot p/(1-p)$ samples of the distribution, and showed that this model smoothly interpolates between a guarantee of $1/e$ and $0.745$, which correspond to the optimal guarantees for the secretary problem and the full-information case.

\subsection{Prophet Secretary and Free Order: the case of known distributions}
Let $n \geq 1$ and $k \geq 1$. 
Consider $n$ independent positive random variables $X_1,X_2,\dots,X_n$, which distributions are known to the gambler. 
The problem proceeds as follows:
\begin{itemize}
\item 
A permutation $\sigma$ is drawn uniformly among the set of permutations of $\left\{1,\dots,n\right\}$,
\item 
At each time $t=1,\dots,n$, the gambler is informed of the realization of $X_{\sigma(t)}$, as well as $\sigma(t)$. He has to choose whether to pick 
$X_{\sigma(t)}$ or not. If he picks it, this is his final reward, and otherwise, we go to stage $t+1$. 
\end{itemize}
The gambler aims at finding a stopping rule $T$ that maximizes $\m{E}(X_T)$. It is well-known that such a maximum can be realized with an \textit{adaptive threshold algorithm}, that is, an algorithm that at each stage makes a decision based on a threshold depending only on the identity of the variables that have arrived so far. Formally,  
an \textit{adaptive threshold algorithm} is a mapping $\pi: \cup_{t=1}^{n} \left\{1,\dots,n\right\}^t \rightarrow \m{R}_+$, with the following interpretation: at stage $t$, if variables $\sigma(1),\dots,\sigma(t-1)$ have been observed, then the gambler picks variable $X_{\sigma(t)}$ if and only if $X_{\sigma(t)}>\pi(\sigma(1),\dots,\sigma(t-1),\sigma(t))$.
In all this paper, we will restrict to adaptive threshold algorithms and \textit{randomized} adaptive threshold algorithms, which correspond to probability distributions over adaptive threshold algorithms. To avoid repetition, we will simply call them  ``algorithm'' and ``randomized algorithm'', respectively. 
\\

If the gambler knew the realizations of the $X_i$ beforehand, he would be able to secure $\m{E}(\max X_i)$. The main question in this problem is what is the maximal constant $C_{S} \in [0,1]$ such that, for any $F_1,\dots,F_n$, there exists a stopping rule $T$ satisfying 
\begin{equation*}
    \m{E}(X_T) \geq C_S \cdot \m{E}(\max X_i).
\end{equation*}
Though such a constant has not been determined yet, it has been shown that $0.669 \leq C_S \leq 0.724$.
\\

The \textit{Free-Order} problem proceeds similarly, to the difference that the permutation $\sigma$ is chosen by the gambler, instead of being drawn uniformly. In this context, a threshold stopping rule can be viewed as a pair $(\sigma,\pi)$, where $\sigma$ is a permutation of $[n]$, and $\pi \in \m{R}_+^n$.
The permutation represents the order of the variables, while $\pi_i$ represents the threshold used at stage $i$. Note that thresholds are assumed to be non-adaptive. This is without loss of generality, since there is no relevant information that the gambler can learn online. Indeed, the order is fixed beforehand by the gambler, and the observed past values are irrelevant, by independence. 

We will call \textit{algorithm} such a stopping rule, and \textit{randomized algorithm} a probability distribution over algorithms. We call $C_F$ the corresponding constant ($F$ stands for ``Free''). It is known that $0.725 \leq C_F \leq 0.745$. 
\subsection{Sample-based Prophet Secretary and Free-Order: the case of unknown distributions}

Let us modify the Prophet Secretary setting described previously by assuming that the gambler does not know the  distributions $F_1,\dots,F_n$, but instead has access to some number of samples for each variable. Formally, let $S^1_i,\dots,S^k_i$ be $k$   independent copies of $X_i$, that we will call \textit{samples} of $X_i$.
Before the game starts, the gambler is informed of the realizations of  samples $S^j_i, i=1 \dots n, \ j=1 \dots k$. Then, the problem proceeds as in the previous setting: the gambler is presented with the $X_1,\dots,X_n$ in random order, and at each step has to decide whether to stop and pick the variable, or discard it and continue. 

The sample-based Free Order model is defined similarly. 
A natural question is then to ask how many samples the gambler needs in order to achieve the same ratio as in the full information case. Our main result is the following: 
\begin{theorem} \label{main_thm}
Assume that $k \geq O(\varepsilon^{-5})$. In the Prophet Secretary model, the gambler can achieve an expected payoff at least equal to $(C_S-\varepsilon) \m{E}(\max_{1 \leq i \leq n} X_i)$. Moreover, we can compute a stopping policy that attains this bound in time polynomial in $n$. 

The same results hold in the Free Order model, replacing $C_S$ by $C_F$.
\end{theorem}
The algorithm that achieves such a ratio is randomized. 
The fact that randomization is required essentially comes from the fact that to an extent the sample-based problem is adversarial: from the samples we can estimate but not exactly calculate the expectation of the maximum or the expectation of what is obtained by the algorithm, and we must be prepared for the worst case over the $n$ distributions.
\\

A characteristic feature of our proof is that it treats Prophet Secretary and Free Order in an almost identical way. To avoid unnecessary repetition, we will specify which of these two problems we are addressing only when some case distinction has to be made.
\\



\section*{Step 1: upper tails distributions do not need to be known}
By a slight abuse, we will use the same notation for an algorithm $ALG$, and the realized payoff it achieves. We will hence denote by $\m{E}_X(ALG)$ its expected payoff, where $X=(X_1,\dots,X_n)$ is the instance under consideration. When there is no ambiguity, we will drop the index $X$ in the expectation. We will also use notation $X^*:=\max_{1 \leq i \leq n} X_i$. In Step 1, all statements are valid both for Prophet Secretary and the Free Order models. The notation $C^*$ stands for the optimal ratio of the model under consideration, that is, $C_S$ for Prophet Secretary, and $C_F$ for Free Order. Fix some $\alpha>0$ and $\varepsilon>0$. 
\\

The goal of this section is to prove the following proposition: 
\begin{proposition} \label{prop:close}
Let $(F_1,\dots,F_n)$ be an instance distribution, and $M_1,\dots,M_n \geq 0$ such that 
$\prod_{i=1}^n F_i(M_i) \geq 1-\alpha$. 
Assume that the gambler has access to $M_1,\dots,M_n$, and to some instance distribution $(F'_1,\dots,F'_n)$ satisfying that for all $i$, for all $x \leq M_i$,
\begin{equation*}
    (1-\varepsilon) (1-F'_i(x)) \leq 1-F_i(x) \leq (1+\varepsilon) (1-F'_i(x)).
\end{equation*}
Then, there exists an algorithm that depends only on $(F'_1,\dots,F'_n)$ such that, if the realizations come from $F_1,\dots,F_n$, the gambler guarantees a ratio $C^*(1-\alpha)(1-\varepsilon)^2$.
\end{proposition}
The above proposition means that, in order to secure a $C^*(1-\alpha)(1-\varepsilon)$ (hence, losing only an $O(\alpha+\varepsilon)$ factor with respect to $C^*$ when $\alpha$ and $\varepsilon$ are small), the gambler only needs to know a ``multiplicative'' $\varepsilon$-approximation of each distribution, and furthermore, does not need to know ``upper tails.''

We start by proving such a proposition for $\varepsilon=0$, namely: 
\begin{proposition} \label{prop:close2}
Let $(F_1,\dots,F_n)$ be an instance, and $M_1,\dots,M_n \geq 0$ such that 
$\prod_{i=1}^n F_i(M_i) \geq 1-\alpha$. 
Assume that the gambler has access to $M_1,\dots,M_n$ and $F_i(x)$, for all $i$ and $x \leq M_i$. 
Then, there exists an algorithm such that when presented with realizations of $F_1,\dots,F_n$, the gambler guarantees a ratio of $(1-\alpha)C^*$.
\end{proposition}
The proof of Proposition \ref{prop:close2} relies on two intermediary results. The first one is a technical lemma, while the second one is a proposition that is of independent interest for the study of Prophet Secretary and Free Order problems. 
\begin{lemma} \label{lem:payoff}
Let $(F_1,\dots,F_n)$ be an instance, and $M_1,\dots,M_n \geq 0$ such that $\prod_{i=1}^n F_i(M_i) \geq 1-\alpha$. 
Let $ALG$ be some algorithm such that for all $i \in [n]$, when $X_i$ is proposed and $X_i > M_i$, then the algorithm picks $X_i$. For each $i \in [n]$, let $A_i$ be the event ``ALG does not stop before variable $X_i$ appears'', $B_i$ be the event $``\forall j \neq i, X_j \leq M_j$'', and
$D$ be the event `` $\exists \ i, \ X_i>M_i$''. Then, 
\begin{equation*}
\m{E}(1_D ALG) \geq (1-\alpha) \sum_{i=1}^n \m{E}(1_{X_i > M_i}X_i) \m{P}(A_i| B_i) 
\end{equation*}
\end{lemma}
\begin{proof}
We have
\begin{eqnarray*}
    \m{E}(1_D ALG)&\geq& 
    \sum_{i=1}^n \m{E}(1_{\left\{X_i > M_i\right\} \cap B_i} ALG)
    \\
&\geq& \sum_{i=1}^n \m{E}(1_{\left\{X_i > M_i\right\} \cap A_i \cap B_i} X_i).
\end{eqnarray*}
Moreover,
\begin{eqnarray*}
    \m{E}(1_{\left\{X_i > M_i \right\} \cap A_i \cap B_i} X_i)&=& \m{E}(1_{\left\{X_i > M_i\right\}} X_i|
    A_i \cap B_i)
    \m{P}(A_i \cap B_i)
    \\
    &\geq& \m{E}(1_{\left\{X_i > M_i\right\}} X_i) \m{P}(A_i| B_i) \m{P}(B_i)
    \\
    &\geq&
    (1-\alpha) \m{E}(1_{\left\{X_i > M_i\right\}} X_i) \m{P}(A_i| B_i).
\end{eqnarray*}
Thus, we get
\begin{eqnarray*}
    \m{E}(1_{D} ALG) \geq
   (1-\alpha) \sum_{i=1}^n \m{E}(1_{X_i > M_i}X_i) \m{P}(A_i|B_i).
\end{eqnarray*}
\end{proof}
\begin{proposition} \label{prop:minmax}
Let $(Y_1,\dots,Y_n)$ be an instance such that all variables are bounded by some $y \in \m{R}_+$. Then, there exists a randomized algorithm that guarantees a ratio $C^*$ for this instance, and that in addition satisfies that for all $i \in \left\{1,\dots,n\right\}$, $\m{P}_Y(A_i) \geq C^*$. 
\end{proposition}
\begin{proof}
Recall that in the Prophet Secretary problem, we consider adaptive threshold algorithms, that correspond to mappings 
 from $\cup_{j=0}^{n} \left\{1,\dots,n\right\}^j$ to $\m{R}_+$. In the Free order problem, an algorithm is a pair $(\sigma,\tau) \in \Sigma_n \times \m{R}_+^n$, where $\sigma$ is a permutation of $\left\{1,\dots,n\right\}$. 
 Because all the $Y_i$ are bounded by $y$, 
 we can assume without loss of generality that all thresholds take values in $[0,y]$. This makes the set of algorithms a compact set, that we denote by $\mathcal{A}$. 
 
Define the zero-sum game where 
Player 1 chooses an algorithm $ALG$ in $\mathcal{A}$, and Player 2 chooses $b \in \m{R}^n_+$. The payoff is
\begin{equation*}
    \gamma(ALG,b)= \m{E}_Y(ALG)+ \sum_{i=1}^n b_i \m{P}_Y(A_i)-C^* \cdot \m{E}(Y^*)-
    C^* \cdot \sum_{i=1}^n b_i.
\end{equation*}
Player 1's action set $\mathcal{A}$ is compact, Player 2's action set $\m{R}^n_+$ is convex, and the payoff function is linear in Player 2's action. In order to apply Sion's minmax theorem, we would need Player 1's action set to be convex, and the payoff function to be linear in Player 1's action. To this aim, we extend the set of actions of Player 1, by considering 
$\mathcal{M}$ the set of probability measures over $\mathcal{A}$. For $\mu \in \mathcal{M}$ and $b \in \m{R}^n_+$, define 
$\gamma^*(\mu,b)$ as being the expectation of $\gamma(ALG,b)$, where ALG is distributed according to $\mu$. The normal-form zero-sum game $(\mathcal{M},\m{R}^n_+,\gamma^*)$ then satisfies all the assumptions of Sion's theorem, hence has a value $v$: 
\begin{equation}
    v=\max_{\mu \in \mathcal{M}} \inf_{b \in \m{R}^n_+}
    \gamma^*(\mu,b)=\inf_{b \in \m{R}^n_+} \max_{\mu \in \mathcal{M}}
    \gamma^*(\mu,b).
\end{equation}
We claim that $v \geq 0$. To this aim, it is enough to show that for any $\delta>0$, for any $b \in \m{R}^n_+$, there exists an algorithm $ALG$ satisfying
$\gamma(ALG,b) \geq -\delta$. Given $\delta>0$ and $b \in \m{R}^n_+$, let $p \in (0,1]$ be small enough so that $[1-(1-p)^{n-1}] C^* \sum_{i=1}^n b_i \leq \delta/2$, $np\cdot \m{E}(Y^*) \leq \delta/2$ and $b_i/p \geq \sum_{i=1}^n b_i+y$. Let $Z_1,\dots,Z_n$ be i.i.d. Bernoulli random variables of parameter $p$. 
Define variables $(Y'_1,\dots,Y'_n)$ by $Y'_i:=Y_i+(b_i/p) \cdot Z_i$.  
\\
We claim that there exists an algorithm $A \in \mathcal{A}$ that guarantees a ratio $C^*$ for the instance $Y'=(Y'_1,\dots,Y'_n)$. This fact is not entirely straightforward, since in $\mathcal{A}$, thresholds are restricted to be below $y$. First, the fact that $b_i/p \geq \sum_{i=1}^n b_i+y$ implies that $b_i/p \geq \m{E}(Y'^*)$. Hence, it is optimal for the gambler to pick any value $Y'_i$ such that $Z_i$ is active. These values are the only ones that are above $y$, and we deduce our claim.  

We can couple the execution of $ALG$ on the instance $Y'$ with its execution on the instance $Y$ by ignoring the term $(b_i/p)\cdot Z_i$ when $Z_i$ is active. Notice that on $Y'$, $ALG$ always stops earlier (or at the same time) as on $Y$. Also, notice that if on $Y'$ it stops earlier, it must be at an element $i$ for which $Z_i$ is active and the algorithm has not stopped yet on $Y$. In that case, the execution on $Y'$ gets $Y_i+b_i/p$, which is at most $Y^*+b_i/p$. Therefore, on $Y'$, $ALG$ gets at most whatever it gets on $Y$, plus $Y^* +b_i/p$ on elements $i$ where $Z_i$ is active and $ALG$ does not stop before $i$ arrives. Since $A_i$ is independent of $Z_1,\dots,Z_n$, we get that
\begin{eqnarray*}
    \m{E}_{Y'}(ALG) &\leq& 
    \m{E}_{Y}(ALG)+\sum_{i=1}^n \frac{b_i}{p}\cdot p \cdot \m{P}_Y(A_i) + np\cdot \m{E}(Y^*)\\
    &\leq &
    \m{E}_{Y}(ALG)+\sum_{i=1}^n \frac{b_i}{p}\cdot p \cdot \m{P}_Y(A_i) + \delta/2.
\end{eqnarray*}
Moreover, by definition of $ALG$, we have
\begin{eqnarray*}
\m{E}_{Y'}(ALG) &\geq& C^* \m{E}(Y'^*)
\\
&\geq& C^* \m{E}(Y^*)+ C^*(1-p)^{n-1}
\sum_{i=1}^n b_i
\\
&\geq& C^* \m{E}(Y^*)+C^* \sum_{i=1}^n b_i-\delta/2. 
\end{eqnarray*}
It follows that $\gamma(ALG,b) \geq -\delta$. Hence, $v \geq 0$. 
\\

Consequently, there exists $ALG$ a randomized algorithm such that for all $b \in \m{R}_+^n$, $\gamma^*(ALG,b) \geq v \geq 0$. Let $i \in [n]$ and $N \geq 1$. Consider $b \in \m{R}^n_+$ defined by $b_i=N$, and $b_j=0$ for $j \neq i$. We have
\begin{equation*}
\gamma^*(ALG,b):=\m{E}_Y(ALG)+ N \m{P}_Y(A_i)-C^* \cdot \m{E}(Y^*)-
    C^* N \geq 0,
\end{equation*}
and taking $N$ to infinity, we deduce that $\m{P}_Y(A_i) \geq C^*$. Hence, the proposition is proved. 
\end{proof}

We are now ready to prove \Cref{prop:close2}. 
\begin{proof}[Proof of Proposition \ref{prop:close2}]
Consider random variables $(Y_1,\dots,Y_n)$ defined by $Y_i:=1_{X_i \leq M_i} X_i, i \in [n]$. By Proposition \ref{prop:minmax}, there exists $ALG$ a randomized algorithm such that $\m{E}_Y(ALG) \geq C^* \cdot \m{E}(Y^*)$ and for all $i \in [n]$, $\m{P}_Y(A_i) \geq C^*$. 
By Lemma \ref{lem:payoff}, we have
\begin{eqnarray*}
    \m{E}_X(ALG) &=&
    \m{E}_X(1_{\left\{\forall i, X_i \leq M_i\right\}} ALG)+ 
    \m{E}(1_{\left\{\exists i, X_i>M_i\right\}} ALG)
    \\
    &\geq& 
   \m{E}_Y(ALG) +(1-\alpha) \sum_{i=1}^n \m{E}(1_{X_i > M_i}X_i) \m{P}_X(A_i|B_i).
\end{eqnarray*}
Since $\m{P}_X(A_i|B_i)=\m{P}_Y(A_i) \geq C^*$, we deduce that
\begin{eqnarray*}
  \m{E}_X(ALG) &\geq& C^* \cdot \m{E}(Y^*)+
    C^*(1-\alpha) \cdot \sum_{i=1}^n \m{E}(1_{X_i > M_i}X_i) 
    \\
    &\geq& 
    C^*(1-\alpha) \cdot \m{E}(Y^*)+
   C^*(1-\alpha) \cdot \m{E}\left(\max_{i \in [n]} \left\{1_{X_i>M_i} X_i\right\}\right)
    \\
    &\geq& C^*(1-\alpha) \m{E}(X^*).
\end{eqnarray*}
 We deduce that ALG guarantees a factor $C^*(1-\alpha)$.
\end{proof}
Let us now proceed with the proof of Proposition \ref{prop:close}. We need first the following lemma: 
\begin{lemma} \label{lem:close2}
Let $ALG$ be some algorithm. There is an algorithm $ALG^*$ such that for any two instances $(F_1,\dots,F_n)$ and $(F'_1,\dots,F'_n)$ that satisfy that for all $x$,
\begin{equation*}
    (1-\varepsilon) (1-F'_i(x)) \leq 1-F_i(x) \leq (1+\varepsilon)(1-F'_i(x)),
\end{equation*}
we have that $\m{E}_X(ALG^*) \geq (1-\varepsilon) \m{E}_{X'}(ALG)$. 
\end{lemma}
This lemma means that if two instances $F$ and $F'$ are $\varepsilon$-close 
``in a multiplicative way'', then we can design an algorithm for $F'$, and the performance of the algorithm against $F$ will be $\varepsilon$-close to the one of the same algorithm against $F'$. Note that if one considers instead an ``additive'' condition, such as $\left|F_i(x)-F'_i(x)\right| \leq \varepsilon$, then the result would not hold (see \cite{DK19}). 
\begin{proof}
We define $ALG^*$ by modifying $ALG$ in the following way: we draw i.i.d. Bernoulli$(\frac{1}{1+\varepsilon})$ random variables $Z_1,\dots,Z_n$, and multiply the $i$-the realization by $Z_i$. If we run $ALG^*$ on realizations of $F_1,\dots,F_n$, the expectation we get is the same as running $ALG$ on realizations drawn from $F^*_1,\dots,F^*_n$ defined by $(1-F^*_i)=\frac{1-F_i}{1+\varepsilon}$ for each $i$. We have that
\[
\frac{1-\varepsilon}{1+\varepsilon} (1-F'_i(x)) \leq 1-F^*_i(x) \leq 1-F'_i(x).
\]

Now we argue about the performance of $ALG$ on both instances by coupling the realizations and the permutation. Since $F'_i$ statistically dominates $F^*_i$ for every $i$, we can couple the realizations $X'_1,\dots,X'_n$ and $X^*_1,\dots, X^*_n$ such that $X'_i\geq X^*_i$ for all $i$ with probability $1$. This means that $ALG$ will always stop later when presented with $X^*_1,\dots, X^*_n$. Finally, conditional on reaching a realization $X^*_i$, $ALG$ obtains from it a reward that is at least a fraction $\frac{1-\varepsilon}{1+\varepsilon}$ of what it obtains from a realization $X'_i$, conditional on reaching it. This stems from the fact that $F^*_i$ approximately statistically dominates $F'_i$, i.e., if we multiply $X'_i$ by a Bernoulli$(\frac{1-\varepsilon}{1+\varepsilon})$, then $X^*_i$ statistically dominates the result. We conclude by noticing that by definition $\m{E}_X(ALG^*)=\m{E}_{X^*}(ALG)$.


\end{proof}
We are now ready to prove the main result of this section. 
\begin{proof}[Proof of Proposition \ref{prop:close}]
Consider $(F_1,\dots,F_n)$, $(F'_1,\dots,F'_n)$ and $M_1,\dots,M_n$ as defined in the statement of Proposition \ref{prop:close}. 
For each $i \in [n]$, define $F''_i$ by $F''_i(x)=F'_i(x)$ if $x \leq M_i$, and $1-F''_i(x)=(1-\varepsilon) (1-F_i(x))$ if $x>M_i$. We have
$\prod_{i=1}^n F''_i(M_i)=\prod_{i=1}^n F_i(M_i) \geq 1-\alpha$, 
and by assumption, the gambler knows $M_i$ and can compute $(F''_1(x),\dots,F''_n(x))$, for all $x \leq M_i$. 
Applying Proposition \ref{prop:close2} to $F''$, there exists an algorithm $ALG$ such that 
$\m{E}_{X''}(ALG) \geq C^*(1-\alpha) \m{E}(X''^*) \geq 
C^*(1-\alpha)(1-\varepsilon) \m{E}(X^*)$. 
Applying Lemma \ref{lem:close2} to $(F_1,\dots,F_n)$ and $(F''_1,\dots,F''_n)$, we get that $\m{E}_{X}(ALG) \geq (1-\varepsilon) \m{E}_{X''}(ALG)\geq (1-\alpha)(1-\varepsilon)^2\m{E}_{X}(ALG)$, and the proposition is proved. 
\end{proof}

\section*{Step 2: small variables can be treated as i.i.d. variables}

Notice that, in any given instance, by replacing any set of distributions with their geometric mean, the distribution of the maximum does not change. In this section, we prove \Cref{prop:eq_iid}, which guarantees that if most random variables are $\varepsilon$-small, by treating these variables as i.i.d. realizations of the geometric mean of their distributions, we do not lose much in the competitive ratio. Recall that a random variable is $\varepsilon$-small if the probability it equals zero is at least $(1-\varepsilon)$.

%

For an instance $I=(F_1,\dots,F_n)$, an arrival order $\sigma$ and an algorithm $ALG$ defined as a sequence of thresholds, we denote by $ALG(I,\sigma)$ the reward obtained from applying $ALG$ to a sequence of variables drawn from $I$ with arrival order $\sigma$, i.e., to a sequence $X_{\sigma(1)},\dots,X_{\sigma(n)}$, where $X_1\sim F_1,\dots, X_n\sim F_n$ independently. We also denote by $I_\sigma$ the instance obtained by reordering $I$ according to $\sigma$.

\begin{proposition}
\label{prop:eq_iid}
    Given an instance $I=(F_1,\dots, F_n)$ and $s\in [n]$, let 
    $G=(\prod_{i=1}^s F_i)^{1/s}$ and define a new instance
    $I'=(G,G,\dots,G,F_{s+1},F_{s+2},\dots,F_n)$. If $s\geq n/2$, and, for a given $\varepsilon>1/\sqrt{n}$, all distributions $F_i$ with $i\leq s$ are $\varepsilon$-small, then for any deterministic algorithm $ALG$ (a sequence of thresholds) and permutation $\sigma$,
    $$
    \m{E}\Big(ALG(I_{\rho},\sigma)\Big) 
    \geq \m{E}\Big(ALG(I',\sigma) \Big) - O(\varepsilon)\cdot \m{E}\left(\max_{1\leq i\leq n} X_i\right),
    $$
    where $\rho$ is an independent random permutation of $[n]$ that restricted to the first $s$ indices is uniformly random, and for indices $i>s$ equals the identity.
\end{proposition}


This proposition implies that it is enough to design an algorithm for the case where the $\varepsilon$-small variables are i.i.d. If the $\varepsilon$-small variables arrive in uniformly random order, then we obtain almost the same expected reward, even if we condition on the arrivals of the large variables. Thus, if the expected reward of the algorithm is a $\gamma$ fraction of the expected maximum when the $\varepsilon$-small variables are i.i.d., its expected reward is at least a $\gamma-O(\varepsilon)$ fraction of the expected maximum when they are not.

\Cref{prop:eq_iid} is stated for deterministic algorithms and permutations to make it easier to read. However, note that by linearity of expectation, the inequality also holds for a randomized algorithm $ALG$ and a random permutation $\sigma$ (possibly non-uniform), even in the case where $ALG$ is arbitrarily correlated with $\sigma$. This means the proposition can be applied to any of the three variants.


Before giving the proof, we restate a useful result from \cite{CSZ21} using our notation and prove two intermediate lemmas. Intuitively, \Cref{lem:CSZ} says that on $I'$, the algorithm stops earlier than on $I$. Next, \Cref{lem:eps_small_ineq} states that the expected reward we get from stopping with a uniformly random $\varepsilon$-small variable is at least a $(1-\varepsilon)$ fraction of the expected reward we get from a variable drawn from the geometric mean. These two lemmas are the main ingredients of the proof of the proposition. Lastly, \Cref{lem:geometric_mean} is a technical result about the geometric mean that allows us to apply \Cref{lem:CSZ} even if we condition on the arrival of one of the $\varepsilon$-small variables.

\begin{lemma}{\cite[Lemma~4.3]{CSZ21}}
    \label{lem:CSZ}
    Given distributions $F_1,\dots,F_s$ and thresholds $\tau_1,\dots,\tau_s$, define $G=\left(\prod_{i=1}^s F_i\right)^{1/s}$. If $\sigma:[s]\rightarrow [s]$ is a uniformly random permutation and $X_i\sim F_i$ for all $i\in [s]$ independently, then for every $k\in [s]$,
    \begin{align*}
        \m{P}( X_{\sigma(i)} \leq \tau_i \text{ for all } i\leq k) \geq  \prod_{i=1}^k G(\tau_i).
    \end{align*}
\end{lemma}

\begin{lemma}
    \label{lem:eps_small_ineq}
    Given $\varepsilon$-small distributions $F_1,\dots,F_s$ and a threshold $\tau$, let $X_1,\dots,X_s\sim F_1\dots,F_s$ independently, and $Y\sim (\prod_{i=1}^s F_i)^{1/s}$. We have that
    \[
    \frac{1}{s}\sum_{i=1}^s \m{E}(X_i\cdot \mathds{1}_{X_i>\tau})\geq (1-\varepsilon)\cdot \m{E}(Y\cdot \mathds{1}_{Y>\tau}).
    \]
\end{lemma}
\begin{proof}
    We can rewrite the left-hand side of the inequality as
    \begin{align}
        &\frac{1}{s}\sum_{i=1}^s \m{E}(X_i\cdot \mathds{1}_{X_i>\tau})\notag \\
        &=
        \frac{1}{s}\sum_{i=1}^s
        \left(
        \tau(1-F_i(\tau))+\int_{\tau}^\infty (1-F_i(x)) \,dx
        \right)\notag \\
        &=\tau\cdot \frac{1}{s}\sum_{i=1}^s\left(1- F_i(\tau) \right) + \int_{\tau}^\infty
        \frac{1}{s}\sum_{i=1}^s (1-F_i(x)) \, dx. \label{eq:eps_small_ineq}
    \end{align}
    Now, from the fact that for $y\in [1-\varepsilon,1]$ it holds that $(1-y)\leq -\log(y)\leq (1+\varepsilon)\cdot (1-y)$, we have that for all $x\geq 0$,
    \begin{align*}
        \frac{1}{s}\sum_{i=1}^s (1-F_i(x))&\geq
        \frac{1}{1+\varepsilon}\cdot \frac{1}{s}\sum_{i=1}^s -\log F_i(x)\\
        &=\frac{1}{1+\varepsilon}\left(-\log\left(\prod_{i=1}^s F_i(x)\right)^{1/s} \right)\\
        &\geq \frac{1}{1+\varepsilon}\left(
        1-\left(\prod_{i=1}^s F_i(x)\right)^{1/s}\right)\\
        &\geq (1-\varepsilon)
        \left(
        1-\left(\prod_{i=1}^s F_i(x)\right)^{1/s}\right).
    \end{align*}
    Replacing this back in \Cref{eq:eps_small_ineq}, we obtain that
    \begin{align*}
        &\frac{1}{s}\sum_{i=1}^s \m{E}(X_i\cdot \mathds{1}_{X_i>\tau})\\
        &\geq
        (1-\varepsilon)\cdot \left( \tau\cdot
        \left(
        1-\left(\prod_{i=1}^s F_i(\tau)\right)^{1/s}\right)
        +\int_{\tau}^\infty
        \left(
        1-\left(\prod_{i=1}^s F_i(x)\right)^{1/s}\right) \, dx \right)\\
        &= (1-\varepsilon)\cdot \m{E}(Y\cdot \mathds{1}_{Y>\tau}),
    \end{align*}
    which concludes the proof of the lemma.
\end{proof}

\begin{lemma}
    \label{lem:geometric_mean}
    If $n\geq 1/\varepsilon^2$ and $s\geq n/2$, then for all $x_1,\dots x_s\in [0,1]$, it holds that
    \[
    \left(\prod_{i=1}^{s-1} x_i\right)^{1/(s-1)}
    \geq
    (1-\varepsilon)\cdot \left(\prod_{i=1}^s x_i \right)^{1/s} - O(\varepsilon/n^2).
    \]
\end{lemma}
\begin{proof}
    We have that since $\frac{1}{s-1}-\frac{1}{s} = \frac{1}{s\cdot (s-1)}$,
    \begin{align*}
        \left(\prod_{i=1}^{s-1} x_i\right)^{1/(s-1)}&\geq
        \left(\prod_{i=1}^{s} x_i\right)^{1/(s-1)}\\
        &=\left(\prod_{i=1}^{s} x_i\right)^{\frac{1}{s-1}-\frac{1}{s}+\frac{1}{s}}\\
        &=\left(\left(\prod_{i=1}^{s} x_i\right)^{1/s}\right)^{1/(s-1)} \cdot \left(\prod_{i=1}^{s} x_i\right)^{1/s}.
    \end{align*}
    To conclude, notice that if
    \[
    \left(\left(\prod_{i=1}^{s} x_i\right)^{1/s}\right)^{1/(s-1)}<(1-\varepsilon),
    \]
    then
    \[
    \left(\prod_{i=1}^{s} x_i\right)^{1/s}<(1-\varepsilon)^{s-1}\leq (1-\varepsilon)^{n/2-1},
    \]
    which is in $O(\varepsilon/n^2)$. \footnote{It is of course much smaller, but this bound will be sufficient for the proof.}
    
\end{proof}

\begin{proof}[Proof of \Cref{prop:eq_iid}]
    We denote by $\tau_1,\dots,\tau_n$ the sequence of thresholds that define $ALG$.
    Let $(X_1,\dots X_n)$ be a sequence drawn from $(F_1,\dots,F_n)$ and $(Y_1,\dots,Y_n)$ a sequence drawn from $(G,\dots,G,F_{s+1},\dots,F_n)$, all independent. By linearity of expectation, we have that
    \begin{align}
    \m{E}\Big(ALG(I_\rho,\sigma)\Big)
    &= \sum_{i=1}^n \m{E}\left(X_{\rho(\sigma(i))}\cdot \mathds{1}_{\{X_{\rho(\sigma(i))}\geq \tau_i\}}\cdot
    \mathds{1}_{\{X_{\rho(\sigma(j))} < \tau_j, \forall j<i  \}}\right). \label{eq:sum_alg}
    \end{align}
    We now analyze separately the terms of this sum that correspond to $\varepsilon$-small variables, i.e., for which $\sigma(i)\leq s$ and the terms that correspond to large variables, i.e., where $\sigma(i)>s$. For $i$ such that $\sigma(i)\leq s$,
    \begin{align}
        &\m{E}\left(X_{\rho(\sigma(i))}\cdot \mathds{1}_{\{X_{\rho(\sigma(i))}\geq \tau_i\}}\cdot
    \mathds{1}_{\{X_{\rho(\sigma(j))} < \tau_j, \forall j<i  \}}\right)
    \notag\\
    &= \frac{1}{s}\sum_{i'=1}^s 
    \m{E}\left(X_{i'}\cdot \mathds{1}_{\{X_{i'}\geq \tau_i\}}\cdot
    \mathds{1}_{\{X_{\rho(\sigma(j))} < \tau_j, \forall j<i  \}}\,\middle|\, \rho(\sigma(i))=i' \right)
    \notag\\
    &=\frac{1}{s}\sum_{i'=1}^s 
    \m{E}\left(X_{i'}\cdot \mathds{1}_{\{X_{i'}\geq \tau_i\}}\right) \cdot \m{P}\left(X_{\rho(\sigma(j))} < \tau_j, \forall j<i \,\middle|\, \rho(\sigma(i))=i' \right).
    \label{eq:small_term_in_expectation}
    \end{align}
    Now, splitting again into small and large variables and applying \Cref{lem:CSZ}, we have that
    \begin{align*}
        &\m{P}\left(X_{\rho(\sigma(j))} < \tau_j, \forall j<i \,\middle|\, \rho(\sigma(i))=i' \right)\\
        &= \m{P}\left(X_{\sigma(j)} < \tau_j, \forall j<i: \sigma(j)>s \right)
        \\
        & \qquad \qquad \cdot
        \m{P}\left(X_{\rho(\sigma(j))} < \tau_j, \forall j<i : \sigma(j)\leq s \,\middle|\, \rho(\sigma(i))=i' \right)
        \\
        &\geq \m{P}\left(X_{\sigma(j)} < \tau_j, \forall j<i: \sigma(j)>s \right)
        \cdot \prod_{j<i:\sigma(j)\leq s} \left( 
        \prod_{j'\leq s: j'\neq i'} F_{j'}(\tau_j)\right)^{1/(s-1)}
        \\
        &= \m{P}\left(X_{\sigma(j)} < \tau_j, \forall j<i: \sigma(j)>s \right)
        \cdot \left(   
        \prod_{j'\leq s: j'\neq i'} \left(\prod_{j<i:\sigma(j)\leq s} F_{j'}(\tau_j)\right) \right)^{1/(s-1)}
        \\
        &\geq
        \m{P}\left(X_{\sigma(j)} < \tau_j, \forall j<i: \sigma(j)>s \right)
        \cdot (1-\varepsilon)\cdot 
        \left(   
        \prod_{j'\leq s} \left(\prod_{j<i:\sigma(j)\leq s} F_{j'}(\tau_j)\right) \right)^{1/s}
        \\
        &\qquad - O(\varepsilon/n^2)
        \\
        &=
        (1-\varepsilon)\cdot \m{P}\left(Y_{\sigma(j)} < \tau_j, \forall j<i  \right) - O(\varepsilon/n^2),
    \end{align*}
    where in the last inequality we applied \Cref{lem:geometric_mean}, and in the last equality we used the definition of $Y$ and $\rho$. Replacing the last inequality back in \Cref{eq:small_term_in_expectation}, and then applying \Cref{lem:eps_small_ineq} we have that
    if $\sigma(i)\leq s$,
    \begin{align*}
        &\m{E}\left(X_{\rho(\sigma(i))}\cdot \mathds{1}_{\{X_{\rho(\sigma(i))}\geq \tau_i\}}\cdot
    \mathds{1}_{\{X_{\rho(\sigma(j))} < \tau_j, \forall j<i  \}}\right)\\
    &\geq 
    \left( (1-\varepsilon)\cdot \m{P}\left(Y_{\sigma(j)} < \tau_j, \forall j<i  \right) - O(\varepsilon/n^2)
    \right) \cdot 
    \frac{1}{s}\sum_{i'=1}^s 
    \m{E}\left(X_{i'}\cdot \mathds{1}_{\{X_{i'}\geq \tau_i\}}\right) \\
    &\geq 
     (1-2\varepsilon)\cdot \m{P}\left(Y_{\sigma(j)} < \tau_j, \forall j<i  \right)\cdot \m{E}\left(Y_{\sigma(i)}\cdot \mathds{1}_{\{Y_{\sigma(i)} \geq \tau_{i}\}}\right) 
      - O(\varepsilon/n^2)
     \cdot \m{E}\left(Y_{\sigma(i)}\right).
    \end{align*}
    Now, consider $i$ such that $\sigma(i)> s$. By the definition of $\rho$, $\rho(\sigma(i))=\sigma(i)$, and by the definition of $Y$, we have that $Y_{\sigma(i)}\sim X_{\sigma(i)}$. Therefore,
    \begin{align}
        &\m{E}\left(X_{\rho(\sigma(i))}\cdot \mathds{1}_{\{X_{\rho(\sigma(i))}\geq \tau_i\}}\cdot
    \mathds{1}_{\{X_{\rho(\sigma(j))} < \tau_j, \forall j<i  \}}\right)
    \notag\\
    &=
    \m{E}\left(X_{\sigma(i)}\cdot \mathds{1}_{\{X_{\sigma(i)}\geq \tau_i\}}\cdot
    \mathds{1}_{\{X_{\rho(\sigma(j))} < \tau_j, \forall j<i  \}} \right)
    \notag\\
    &=
    \m{E}\left(Y_{\sigma(i)}\cdot \mathds{1}_{\{Y_{\sigma(i)}\geq \tau_i\}}\right)\cdot \m{P}\left(
    X_{\rho(\sigma(j))} < \tau_j, \forall j<i   \right).
    \label{eq:large_term}
    \end{align}
    Like before, we split $j$ into small and large variables, and then apply \Cref{lem:CSZ} to obtain that
    \begin{align*}
        &\m{P}\left(
    X_{\rho(\sigma(j))} < \tau_j, \forall j<i   \right)\\
    &= \m{P}\left(
    X_{\sigma(j)} < \tau_j, \forall j<i : \sigma(j)>s  \right) \cdot
    \m{P}\left(
    X_{\rho(\sigma(j))} < \tau_j, \forall j<i  : \sigma(j)\leq s \right)
    \\
    &\geq
    \m{P}\left(
    X_{\sigma(j)} < \tau_j, \forall j<i : \sigma(j)>s  \right) \cdot
    \prod_{j<i: \sigma(j)\leq s} G(\tau_j)
    \\
    &=
    \m{P}(Y_{\sigma(j)} < \tau_j, \forall j<i).
    \end{align*}
    Replacing this back into \Cref{eq:large_term}, we have that if $\sigma(i)>s$,
    \begin{align*}
        &\m{E}\left(X_{\rho(\sigma(i))}\cdot \mathds{1}_{\{X_{\rho(\sigma(i))}\geq \tau_i\}}\cdot
    \mathds{1}_{\{X_{\rho(\sigma(j))} < \tau_j, \forall j<i  \}}\right)
    \\
    & \qquad
    \geq 
    \m{E}\left(Y_{\sigma(i)}\cdot \mathds{1}_{\{Y_{\sigma(i)}\geq \tau_i\}}\right)\cdot \m{P}\left(
    Y_{\sigma(j)} < \tau_j, \forall j<i   \right).
    \end{align*}
    Putting everything back in \Cref{eq:sum_alg}, we get that
    \begin{align*}
        \m{E}\Big(ALG(I_\rho,\sigma)\Big)
        &\geq 
        (1-2\varepsilon)\cdot \sum_{i=1}^n 
        \m{E}\left(Y_{\sigma(i)}\cdot \mathds{1}_{\{Y_{\sigma(i)}\geq \tau_i\}}\right)\cdot \m{P}\left(
    Y_{\sigma(j)} < \tau_j, \forall j<i   \right)
    \\
    &\qquad - O(\varepsilon/n^2)\sum_{i=1}^n \m{E}(Y_i)\\
    &\geq \m{E}\Big(ALG(I',\sigma)\Big) - O(\varepsilon)\cdot  \m{E}\left(\max_{1\leq i\leq n} Y_i \right).
    \end{align*}
    This concludes the proof of the proposition, as $\max_{1\leq i\leq n} X_i$ and $\max_{1\leq i\leq n} Y_i$ are identically distributed.
\end{proof}

\section*{Step 3: sample-based approximation of the distributions} 
Equipped with the machinery developed in the first two steps, let us go back to our main goal: proving Theorem \ref{main_thm}. In this step, we focus on proving that $k \geq O(\varepsilon^{-5})$ is enough to build an algorithm that guarantees a ratio $C^*-\varepsilon$. The fact that it can found in polynomial time in $n$ will be done in the last step.

Recall that the gambler faces an instance $X_1,\dots,X_n$, with unknown distributions $F_1,\dots,F_n$, and has access to samples $S^j_i, i \in [n], j \in [k]$. We fix some $\varepsilon>0$. 
The first sub-step is to detect variables that are not $\varepsilon$-small. The second sub-step is to estimate distributions of an auxiliary instance where all $\varepsilon$-small variables have been replaced by an i.i.d. distribution, and apply Proposition \ref{prop:close}. The proof then stems from Proposition \ref{prop:eq_iid}. 

Throughout this step, we assume the distributions are all absolutely continuous. If this is not the case, we can always approximate the distributions with absolutely continuous distributions. One way to do this is to first draw a single sample of $X^*$ and then add i.i.d. noise drawn from a Uniform$[0,\delta]$ to every subsequent sample and every realization, where $\delta=\varepsilon^2\cdot X^*$. By Markov's inequality, the probability that $X^*>\m{E}(X^*)/\varepsilon$ is at most $\varepsilon$. Therefore, with probability at least $(1-\varepsilon)$, the extra noise is only an $\varepsilon$ fraction of the expected maximum.

\subsection*{Classification of variables}
Let $T^*$ be such that $\prod_{i=1}^n F_i(T^*) = \varepsilon$. We show in the following lemma that we can basically ignore all values below $T^*$.
\\
\begin{lemma} \label{lem:lowertail}
    If $\prod_{i=1}^n F_i(T^*) = \varepsilon$, then
    \[
    (1-\varepsilon) \cdot \m{E}\left(\max_{1\leq i\leq n} X_i \right) \leq \m{E} \left(\max_{1\leq i\leq n} X_i\cdot \mathds{1}_{\{T^*<X_i\}}\right). 
    \]
\end{lemma}
\begin{proof}
    We split the expectation into values below $T^*$ and values above $T^*$.
    \begin{align*}
        \m{E}\left(\max_{1\leq i\leq n} X_i \right)
        &= \m{E} \left(\max_{1\leq i\leq n} X_i\cdot \mathds{1}_{\{X_i\leq T^*\}}\right)
        + \m{E} \left(\max_{1\leq i\leq n} X_i\cdot \mathds{1}_{\{T^*<X_i\}}\right)\\
        &= \m{E} \left(\max_{1\leq i\leq n} X_i \,\middle|\,   \max_{1\leq i\leq n} X_i\leq T^*\right)\cdot \m{P}\left(\max_{1\leq i\leq n} X_i\leq T^*\right) \\
        & \qquad + \m{E} \left(\max_{1\leq i\leq n} X_i\cdot \mathds{1}_{\{T^*<X_i\}}\right)\\
        &\leq \varepsilon\cdot \m{E}\left(\max_{1\leq i\leq n} X_i \right)
        +\m{E} \left(\max_{1\leq i\leq n} X_i\cdot \mathds{1}_{\{T^*<X_i\}}\right).
    \end{align*}
    Rearranging the terms concludes the proof of the lemma.
\end{proof}

Using $k_1$ samples of each distribution, we get $k_1$ samples of the distribution of the maximum, i.e., of  $\prod_{i=1}^n F_i$. We estimate $T^*$ from below with the $\lfloor k_1\cdot\varepsilon\cdot  (1-\varepsilon)\rfloor$-th smallest of the $k_1$ samples of the maximum and call this quantity $T$. The following lemma establishes how large we have to set $k_1$ so that $T$ is a good estimate of $T^*$.

\begin{lemma}
    \label{lem:correct_T}
    If $k_1\geq 6(1/\varepsilon)^3\cdot \log(1/\varepsilon)$, then with probability at least $1-\varepsilon$,
    \[
    (1-\varepsilon)^2\cdot \varepsilon \leq \prod_{i=1}^n F_i(T) \leq
     \varepsilon.
    \]
\end{lemma}
\begin{proof}
    Denote by $T_{\textsc{low}}$ the value such that $\prod_{i=1}^n F_i(T_{\textsc{low}}) = (1-\varepsilon)^2\cdot \varepsilon$. The statement of the lemma is equivalent to
    \begin{align*}
        \m{P}(T_{\textsc{low}}\leq T \leq T^*) \geq 1-\varepsilon.
    \end{align*}
    Let $M_{\textsc{low}}$ be the number of samples of the maximum below $T_{\textsc{low}}$, and $M^*$ be the number of samples below $T^*$. We have that
    \begin{align*}
        \m{P}(T_{\textsc{low}}\leq T \leq T^*)&=\m{P}(M_{\textsc{low}} \leq \lfloor k_1\varepsilon (1-\varepsilon)\rfloor
        \leq M^*)\\
        &\geq 1- \m{P}(M_{\textsc{low}}>k_1 \varepsilon (1-\varepsilon))
        - \m{P}(M^*< k_1 \varepsilon (1-\varepsilon)).
    \end{align*}
    Since $\m{E}(M_{\textsc{low}})=k_1\varepsilon(1-\varepsilon)^2$ and $\m{E}(M^*)= k_1\varepsilon$, Chernoff bounds imply that
    \begin{align*}
        \m{P}(T_{\textsc{low}}\leq T \leq T^*)&\geq
        1-e^{-\varepsilon^2 k_1 \varepsilon/(2+\varepsilon)} -e^{-\varepsilon^2k_1\varepsilon/2}\\
        &\geq 1-2\varepsilon^2,
    \end{align*}
    which for small $\varepsilon$ is at least $1-\varepsilon$.
\end{proof}

Let $L^*=\{i: 1-F_i(T^*) > \varepsilon \}$. We draw $k_2$ fresh samples of each distribution, and for each $i=1,\ldots,n$ denote as $\hat{F}_i$ the empirical distribution that results from them. Let $L=\{i: 1-\hat{F}_i(T)>(1-\varepsilon)\cdot \varepsilon \}$, and denote by $A$ the event associated with \Cref{lem:correct_T}.

\begin{lemma}
    If $k_2\geq 6(1/\varepsilon)^3\cdot \log(1/\varepsilon)$, then conditional on event $A$, with probability at least $(1-\varepsilon)$,
    \[
    L^* \subseteq L, \text{ and } |L|\leq O((1/\varepsilon)\log(1/\varepsilon)).
    \]
\end{lemma}
\begin{proof}
    We first bound $|L^*|$. By the definition of $L^*$, $F_i(T^*)<1-\varepsilon$ for all $i\in L^*$. Therefore,
    \begin{align}
    \varepsilon=\prod_{i=1}^n F_i(T^*)\leq \prod_{i\in L^*}F_i(T^*)<(1-\varepsilon)^{|L^*|}.
    \label{eq:bound-ellstar-1}
    \end{align}
    Taking the logarithm on both sides and rearranging the terms, we obtain that
    \begin{align}
            |L^*|< \frac{\log(1/\varepsilon)}{\log(1/(1-\varepsilon))}\leq (1/\varepsilon)\log(1/\varepsilon).
            \label{eq:bound-ellstar-2}
    \end{align}

    We show now that, conditional on $A$, $L^*\subseteq L$ with probability at least $1-\varepsilon/2$. Take $i\in L^*$ and define the random variable $Y_i$ as the number of samples of $F_i$ that are larger than $T^*$. Conditional on $A$, $T\leq T^*$, so the definition of $L$ implies that if $Y_i> k_2\cdot (1-\varepsilon)\cdot \varepsilon$, then $i\in L$. But since $i\in L^*$, we have that $\m{E}(Y_i)> k_2\cdot \varepsilon$. 
    A simple Chernoff bound implies that
    \[
    \m{P}(Y_i\leq k_2\cdot (1-\varepsilon)\cdot \varepsilon) \leq e^{-\varepsilon^2\cdot k_2 \cdot \varepsilon/2} \leq \varepsilon^3.
    \]
    Taking a union bound over the elements of $L^*$, we conclude that $L^*\subseteq L$ with probability at least $1-\varepsilon^2\log(1/\varepsilon)\geq 1-\varepsilon/2$.
    
    Denote by $T_{\textsc{low}}$ the value such that $\prod_{i=1}^n F_i(T_{\textsc{low}}) = (1-\varepsilon)^2\cdot \varepsilon$. Notice that the event $A$ is exactly the event that $T_{\textsc{low}}\leq T\leq T^*$. Define the set $L_{\textsc{low}}=\{i: 1-\hat{F}_i(T_{\textsc{low}})>(1-\varepsilon)\cdot \varepsilon \}$. The event $A$ implies that $L\subseteq L_{\textsc{low}}$, so to conclude it is enough to bound $|L_{\textsc{low}}|$. 
    Denote by $Z_i$ the number of samples from $F_i$, out of the $k_2$, that are larger than $T_{\textsc{low}}$. We have that $i\in L_{\textsc{low}}$ if and only if $Z_i>k_2\cdot (1-\varepsilon)\cdot \varepsilon$. Therefore,
    \[
    |L_{\textsc{low}}|\leq \frac{1}{k_2\cdot (1-\varepsilon)\cdot \varepsilon} \sum_{i=1}^n Z_i.
    \]
    Notice that
    \[
    \m{E}\left(\sum_{i=1}^n Z_i \right) = k_2 \cdot \sum_{i=1}^n (1-F_i(T_{\textsc{low}})).
    \]
    From the definition of $T_{\textsc{low}}$ we have that
    \begin{align*}
        \log\left( \frac{1}{(1-\varepsilon)^2\cdot \varepsilon} \right) = \sum_{i=1}^n -\log(F_i(T_{\textsc{low}}))\geq \sum_{i=1}^n (1-F_i(T_{\textsc{low}})).
    \end{align*}    
    Thus, a Chernoff bound implies that
    \begin{align*}
        \m{P}\left( 
        \sum_{i=1}^n Z_i > 6 k_2 \cdot \log(1/\varepsilon)
        \right)
        \leq e^{-2^2 k_2 \log(1/\varepsilon) /(2+2)}
        \leq \varepsilon/2.
    \end{align*}
    Therefore, $|L_{\textsc{low}}|\leq O( (1/\varepsilon)\log(1/\varepsilon))$ with probability at least $1-\varepsilon/2$.
\end{proof}
\subsection*{Estimating the auxiliary instance}
The previous sub-step tells us that for $\varepsilon$ small enough, by drawing no more than $\varepsilon^{-4}$ samples, one can construct a (random) number $T$ and a (random) subset $L \subset [n]$ such that with probability higher than $1-\varepsilon$, the following holds: 
\begin{enumerate}
\item 
$(1-\varepsilon) \cdot \m{E}\left(\max_{1\leq i\leq n} X_i \right) \leq \m{E} \left(\max_{1\leq i\leq n} X_i\cdot \mathds{1}_{\{T<X_i\}}\right)$
\item
$\prod_{i \in [n]} F_i(T) \geq (1-\varepsilon)^2 \varepsilon$
\item
$|L| \leq \varepsilon^{-2}$
\item 
For all $i \in S:=[n] \setminus L$, $1-F_i(T) \leq \varepsilon$
\end{enumerate}
Hence, up to considering $Y_i:=X_i 1_{X_i>T}$, one can assume without loss of generality that all variables $X_i, i \in S$ are $\varepsilon$-small, and that moreover, for all $x$, $\prod_{i \in [n]} F_i(x) \geq (1-\varepsilon)^2 \varepsilon$. 
Using $\varepsilon^{-4}$ more samples for each variable, an analogous argument as in Lemma \ref{lem:correct_T} allows to construct a random number $M$ satisfying that with probability at least $1-\varepsilon$,  
$(1-\varepsilon)^3 \leq \prod_{i \in S} F_i(M) \leq 1-\varepsilon$. 
\\

Let $H:=\prod_{i \in S} F_i$ be the cumulative distribution of $\max_{i \in S} X_i$, and $G:=H^{\frac{1}{|S|}}$. Let $\hat{H}$ be the empirical cumulative distribution of $\max_{i \in S} X_i$, obtained by considering another independent set of $k \geq \varepsilon^{-5}$ samples of each $X_i$. 
Let $\hat{G}:=\hat{H}^{1/|S|}$. 
\begin{proposition}
The following statement holds with probability larger than $1-\varepsilon$: for all $x \leq M$, 
    \begin{equation*}
    (1-\varepsilon) (1-\hat{G}(x)) \leq 1-G(x) \leq (1+\varepsilon) (1-\hat{G}(x)).
\end{equation*}
\end{proposition}
\begin{proof}
By the DKW inequality, with probability larger than $1-2 k^{-2} \geq 1-\varepsilon$, we have 
\begin{equation*}
    \left\|H-\hat{H} \right\|_\infty \leq \left(\frac{\ln(k)}{k} \right)^{1/2} \leq \varepsilon^2/4.
\end{equation*}
Conditional on this event, since $H(x) \geq \prod_{i \in [n]} F_i(x) \geq (1-\varepsilon)^2 \varepsilon$, the above inequality implies $\hat{H}(x) \geq \varepsilon-\varepsilon^2/4$. 
By the Mean Value theorem, we have for all $x$,
\begin{eqnarray*} \label{eq:estim_iid}
\left|\hat{G}(x)-G(x)\right| &\leq&
\sup_{t \in [H(x),\hat{H}(x)]} t^{\frac{1}{|S|}-1} |S|^{-1} \left\|H-\hat{H} \right\|_\infty
\\
&\leq&
\max(H(x)^{-1},\hat{H}(x)^{-1}) |S|^{-1} \left\|H-\hat{H} \right\|_\infty
\\
&\leq&
\varepsilon |S|^{-1}/2. 
\end{eqnarray*}
Take $x \leq M$. Then, $H(x) \leq H(M) \leq 1-\varepsilon$, hence $G(x) \leq (1-\varepsilon)^{\frac{1}{|S|}} \leq 1-\varepsilon |S|^{-1}$, and $\hat{G}(x) \leq 1-\varepsilon |S|^{-1}/2$. 
We deduce that 
\begin{eqnarray*}
    1-G(x) &\leq& 1-\hat{G}(x)+\varepsilon^{2} |S|^{-1}/2
    \\
    &\leq& (1+\varepsilon) (1-\hat{G}(x)),
\end{eqnarray*}
and similarly, 
\begin{equation*}
    1-G(x) \geq (1-\varepsilon) (1-\hat{G}(x)).
\end{equation*}
Hence, 
\begin{equation*}
 (1-\varepsilon) (1-\hat{G}(x)) \leq 1-G(x) \leq (1+\varepsilon) (1-\hat{G}(x))
\end{equation*}
\end{proof}
Let us now estimate the variables in $L$, using another set of $\varepsilon^{-5}$ samples for each variable, and considering the empirical distributions $\hat{F}_i$, $i \in L$. Because $|L| \leq \varepsilon^{-2}$, the multivariate DKW inequality gives that with probability higher than $1-\varepsilon$, for all $i \in L$, $\left\|\hat{F}_i -F_i \right\|_\infty \leq \varepsilon^{9/4}$. Moreover, another set of $\varepsilon^{-5}$ fresh samples allows to compute $M_i, i \in L$ such that with probability $1-\varepsilon$, for all $i \in L$, $(1-\varepsilon^3)^3 \leq F_i(M_i) \leq 1-\varepsilon^3$. We deduce that 
\begin{equation*}
 (1-\varepsilon) (1-\hat{F_i}(x)) \leq 1-F_i(x) \leq (1+\varepsilon) (1-\hat{F_i}(x)). 
\end{equation*}
Since $G(M)^{|S|} \prod_{i \in L} F_i(M_i) \geq (1-\varepsilon)^3(1-\varepsilon^{3})^{\varepsilon^{-2}} \geq 1-O(\varepsilon)$, we are in position to apply Proposition \ref{prop:close} to the instance composed with variables $F_i, i \in L$ and $|S|$ i.i.d. copies of $G$. This gives the existence of an algorithm that guarantees a factor $C^*-O(\varepsilon)$. By Proposition \ref{prop:eq_iid}, the same algorithm guarantees a factor $C^*-O(\varepsilon)$ when presented with realizations of $F_1,\dots,F_n$, and the first part of Theorem \ref{main_thm} is proved. 


\section*{Step 4: Polynomial time computation}

In the previous sections, we showed that there is a strategy that guarantees a $C^*$-approximation using a constant number of samples per distribution. In this section, we complement our main result, proving the following proposition that states that we can compute such a strategy in polynomial time. 

\begin{proposition}
    \label{prop:polytime}
    For an instance $(G,G,\dots,G,F_{s+1},\dots,F_n)$, if $(n-s)$ is bounded by a constant, and the size of the support of each distribution is polynomial in $n$, then we can find in polynomial time in $n$ an algorithm with expected reward at least $C^*\cdot \m{E}(\max_{i\in [n]} X_i)$ and such that for all $i\in [n]$, $\m{P}(A_i)\geq C^*$, where $A_i$ is the event that the algorithm observes $X_i$ before stopping.
\end{proposition}

Notice that such an algorithm is guaranteed to exist by \Cref{prop:minmax}. Notice also that the instance for which we need to compute an algorithm satisfies the assumptions of \Cref{prop:polytime}, as we can replace all $\varepsilon$-small variables with i.i.d. random variables, and we use the empirical distributions, which are supported on the polynomially many samples.

We first introduce linear program formulations that capture the algorithms satisfying the conditions of the proposition for the prophet-secretary and the free-order variants in the case where all distributions have finite support. However, these linear programs have exponential size, as we need to model every possible arrival order. Using the assumption that almost all variables are i.i.d., we can reduce the state space and obtain linear programs of polynomial size.

For each $i\in [n]$, let $V_i$ be a set of indices, $\{x_{i,j}:j\in V_i\}$ the support of distribution $F_i$, and $p_{i,j}$ the probability that a variable drawn from $F_i$ equals $x_{i,j}$. The following linear program captures the algorithm guaranteed to exist by \Cref{prop:minmax} for the prophet-secretary variant.

\begin{align*}
    \text{(PSLP)}\quad &\max_{\alpha,\beta,\gamma,\delta}  {}\,\delta&
    \\
    &\text{s.t. } &
    \\
    &\delta 
    \leq \sum_{S\subseteq [n]: i\in S} \alpha_{S}\cdot \frac{1}{|S|}, 
    & \forall i\in [n]
    \\
    &\delta\cdot \m{E}\left(\max_{i\in [n]} X_i\right) 
    \leq \sum_{S\subseteq [n]}\sum_{i\in S}\sum_{j\in V_i} \beta_{i,j,S}\cdot x_{i,j}
    &
    \\
    &\alpha_{[n]} 
    = 1 &
    \\
    &\alpha_{S} 
    = \sum_{i\in [n]\setminus S} \sum_{j\in V_i} \gamma_{i,j,S\cup \{i\}}, 
    &\forall S\subsetneq [n]
    \\
    &\beta_{i,j,S} + \gamma_{i,j,S} 
    = \alpha_{S}\cdot \frac{1}{|S|} \cdot p_{i,j}, 
    &\forall S\subseteq [n], i\in S, j\in V_i
    \\
    &\alpha_S, \beta_{i,j,S}, \gamma_{i,j,S}, \delta 
    \in [0,1], 
    &\forall S\subseteq [n], i\in S, j\in V_i.
\end{align*}
In this linear program, the variable $\delta$ is the guarantee of the algorithm. Thus, by \Cref{prop:minmax}, $\delta\geq C^*$. For $S\in [n]$, during the execution of the algorithm, we say it is in state $S$ if it has not stopped yet and the set of variables it has not observed yet is exactly $S$. The variable $\alpha_S$ is the probability that the algorithm reaches state $S$ at some point in its execution. For $i\in S$ and $j\in V_i$, the variable $\beta_{i,j,S}$ is the probability that the algorithm reaches state $S$, then it observes variable $i$ with realization $x_{i,j}$, and stops. The variable $\gamma_{i,j,S}$ is the probability that the algorithm reaches the same situation but does not stop. The first two constraints are the conditions of \Cref{prop:minmax}. The third, fourth, and fifth constraints ensure that the variables are consistent with their interpretations as probabilities.

For the free-order variant, we can write an analogous linear program. We denote by $\Sigma_n$ the set of permutations of $[n]$. 
\begin{align*}
    \text{(FOLP)}\quad &\max_{\alpha,\beta,\gamma,\delta}  {}\,\delta&
    \\
    &\text{s.t. } &
    \\
    &\delta 
    \leq \sum_{\sigma\in\Sigma_n} \alpha_{\sigma^{-1}(i),\sigma}, 
    & \forall i\in [n]
    \\
    &\delta\cdot \m{E}\left(\max_{i\in [n]} X_i\right) 
    \\
    &\qquad \leq \sum_{\sigma\in \Sigma_n}\sum_{i\in [n]}\sum_{j\in V_{\sigma(i)} }\beta_{i,j,\sigma}\cdot x_{\sigma(i),j}
    &
    \\
    &\sum_{\sigma\in \Sigma_n} \alpha_{1,\sigma} 
    = 1 &
    \\
    &\alpha_{i,\sigma} 
    = \sum_{j\in V_{\sigma(i-1)}} \gamma_{i-1,j,\sigma}, 
    & \forall i\geq 2, \sigma\in \Sigma_n
    \\
    &\beta_{i,j,\sigma} + \gamma_{i,j,\sigma} 
    = \alpha_{i,\sigma} \cdot p_{\sigma(i),j}, 
    &\forall \sigma\in\Sigma_n, i\in [n], j\in V_{\sigma(i)}
    \\
    &\alpha_{i,\sigma}, \beta_{i,j,\sigma}, \gamma_{i,j,\sigma}, \delta 
    \in [0,1], 
    &\forall \sigma\in \Sigma_n, i\in [n], j\in V_{\sigma(i)}.
\end{align*}
The variables of this program have an analogous interpretation as in the previous one. The only difference is that here, the algorithm first chooses an arrival order $\sigma\in \Sigma_n$, and then follows that order. Therefore, the state space is given by the pair $(i,\sigma)$, which means that the algorithm chose the order given by $\sigma$, and observes the $i$-th variable before stopping.

With the given interpretation of the variables of the linear programs, it is not hard to see that every algorithm has a corresponding feasible solution, and every feasible solution has a corresponding algorithm. We now must argue that when all but a constant number of distributions are identical, we can reduce the state space to have polynomial support. 

If the first $s$ distributions have the same distribution, that means that $V_i=V_{i'}$ for every $i,i'\leq s$, and that for every $j\in V_i$, $x_{i,j}=x_{i',j}$ and $p_{i,j}=p_{i',j}$. This implies that the linear programs are symmetric on indices $i\leq s$, and therefore, since we can relabel the variables and then average, there must be a symmetric solution. Thus, we can write a program that contains only symmetric solutions by replacing all ``repeated" variables with a single one. 

 In the following reduced linear program, we take as state space the subsets of the multiset $U$ that contains $[n]\setminus [s]$ and $s$ copies of $1$. We denote as $m_S(i)$ the multiplicity of $i$ in $S$, and by Supp$(S)$ the set of distinct elements in $S$. Notice that if $(n-s)$ is bounded by a constant, then the number of different subsets of $U$ is bounded by a polynomial in $n$, so we obtain a linear program of polynomial size.

\begin{align*}
    \text{(rPSLP)}\quad &\max_{\alpha,\beta,\gamma,\delta}  {}\,\delta&
    \\
    &\text{s.t. } &
    \\
    &\delta 
    \leq \frac{1}{m_U(i)} \cdot \sum_{S\subseteq U: i\in \text{Supp}(S)} \alpha_{S}\cdot \frac{m_S(i)}{|S|}, 
    & \forall i\in U
    \\
    &\delta\cdot \m{E}\left(\max_{i\in [n]} X_i\right) &
    \\
    & \qquad
    \leq \sum_{S\subseteq U}\sum_{i\in \text{Supp}(S)}\sum_{j\in V_i} \beta_{i,j,S}\cdot x_{i,j}
    &
    \\
    &\alpha_{U} 
    = 1 &
    \\
    &\alpha_{S} 
    = \sum_{i\in \text{Supp}(U\setminus S)} \sum_{j\in V_i} \gamma_{i,j,S+\{i\}}, 
    &\forall S\subsetneq U
    \\
    &\beta_{i,j,S} + \gamma_{i,j,S} 
    = \alpha_{S}\cdot \frac{m_S(i)}{|S|} \cdot p_{i,j}, 
    &\forall S\subseteq U, i\in \text{Supp}(S), j\in V_i
    \\
    &\alpha_S, \beta_{i,j,S}, \gamma_{i,j,S}, \delta 
    \in [0,1], 
    &\forall S\subseteq U, i\in \text{Supp}(S), j\in V_i.
\end{align*}

Similarly, to obtain a reduced version of (FOLP), we take as state space the set of orderings of $U$, that is, the set of functions $\sigma:[n]\rightarrow \text{Supp}(U)$ such that $|\sigma^{-1}(i)|=m_U(i)$ for all $i\in \text{Supp}(U)$. We denote this set as $\Sigma_U$. Notice that $\Sigma_U$ has polynomially many elements.
\begin{align*}
    \text{(rFOLP)}\quad &\max_{\alpha,\beta,\gamma,\delta}  {}\,\delta&
    \\
    &\text{s.t. } &
    \\
    &\delta 
    \leq \frac{1}{m_U(i)}\sum_{\sigma\in\Sigma_U} \sum_{t\in [n]:\sigma(t)=i} \alpha_{t,\sigma}, 
    & \forall i\in \text{Supp}(U)
    \\
    &\delta\cdot \m{E}\left(\max_{i\in [n]} X_i\right) 
    \\
    &\qquad \leq \sum_{\sigma\in \Sigma_U}\sum_{i\in [n]}\sum_{j\in V_{\sigma(i)} }\beta_{i,j,\sigma}\cdot x_{\sigma(i),j}
    &
    \\
    &\sum_{\sigma\in \Sigma_U} \alpha_{1,\sigma} 
    = 1 &
    \\
    &\alpha_{i,\sigma} 
    = \sum_{j\in V_{\sigma(i-1)}} \gamma_{i-1,j,\sigma}, 
    & \forall i\geq 2, \sigma\in \Sigma_U
    \\
    &\beta_{i,j,\sigma} + \gamma_{i,j,\sigma} 
    = \alpha_{i,\sigma} \cdot p_{\sigma(i),j}, 
    &\forall \sigma\in\Sigma_U, i\in [n], j\in V_{\sigma(i)}
    \\
    &\alpha_{i,\sigma}, \beta_{i,j,\sigma}, \gamma_{i,j,\sigma}, \delta 
    \in [0,1], 
    &\forall \sigma\in \Sigma_n, i\in [n], j\in V_{\sigma(i)}.
\end{align*}
\section*{Concluding remarks}
The proof adapts straightforwardly to the i.i.d. model, showing that $O(1/\varepsilon^5)$ samples are good enough to guarantee the constant $0.745-\varepsilon$. Even more, in a non-i.i.d. instance, if all variables are $\varepsilon$-small, $O(1/\varepsilon^5)$ samples are also enough to guarantee the constant $0.745-\varepsilon$ in the prophet-secretary variant. This comes from the fact that the optimal guarantee for this type of instances is $0.745-\varepsilon$ in the full-information case, which was proved by Liu et al.~\cite{LPP21}, but is also a consequence of Step 2. From here, it is easy to conclude the claim that our approach works for the i.i.d. case: when we truncate the distributions in Step 3, at most a constant number of them can be large, which implies that they are all $\varepsilon$-small (because they are i.i.d.). 

Another exciting direction is to modify Step 1 to apply it to other online selection models. The fact that the same technique applies to different well-known models gives promising perspectives on extending our result to multi-choice models, such as matroids, or combinatorial auctions.

Lastly, a surprising observation is that a result like ours is impossible if we want to approximate the optimal online algorithm. Consider the following example: all variables are $0$ with probability $(1-\varepsilon)$ and $1/\varepsilon$ with probability $\varepsilon$, except for one, which is $0$ with probability $1-1/e^n$, and $n\cdot e^n$ with probability $1/e^n$. Almost all the value comes from this last variable, so the optimal online algorithm will wait to see it before stopping, but even a polynomial number of samples is not enough to identify it. 
\section*{Acknowledgments}
The authors are grateful to Jose Correa for valuable discussions that helped improve this paper.
This work was supported by the French Agence Nationale de la Recherche (ANR) under reference ANR-21-CE40-0020 (CONVERGENCE project).
\bibliographystyle{plain}
\bibliography{bibliography}

\end{document}